%% file: main.tex
\documentclass[runningheads]{llncs}
\usepackage{graphicx}
\usepackage{caption}
\usepackage[utf8x]{inputenc} 
\usepackage{amsmath} 
\usepackage[noabbrev]{cleveref}
\usepackage{url}
\usepackage{mathtools}
\usepackage{amssymb}
\usepackage{placeins}
\usepackage{algorithm,algpseudocode}

\DeclarePairedDelimiter{\abs}{\lvert}{\rvert}

\usepackage{crunch}
\usepackage{tikz}

\newcommand{\Ripser}{\footnote{\url{https://github.com/Ripser/ripser}}~}

\begin{document}
\mainmatter  
\title{$\epsilon$-net Induced Lazy Witness Complexes\\ on Graphs}

\author{%
	Naheed Anjum Arafat \inst{1},
	Debabrota Basu \inst{2},
	St\'ephane Bressan \inst{1}
}%
\institute{
	School of Computing, National University of Singapore, Singapore \and Data Science and AI Division, Chalmers University of Technology, Sweden
}

\maketitle

\begin{abstract}
Computation of persistent homology of simplicial representations such as the Rips and the C\v{e}ch complexes do not efficiently scale to large point clouds. It is, therefore, meaningful to devise approximate representations and evaluate the trade-off between their efficiency and effectiveness. The lazy witness complex economically defines such a representation using only a few selected points, called landmarks.

Topological data analysis traditionally considers a point cloud in a Euclidean space. In many situations, however, data is available in the form of a weighted graph. A graph along with the geodesic distance defines a metric space. 
This metric space of a graph is amenable to topological data analysis.

We discuss the computation of persistent homologies on a weighted graph. We present a lazy witness complex approach leveraging the notion of $\epsilon$-net that we adapt to weighted graphs and their geodesic distance to select landmarks. We show that the value of the $\epsilon$ parameter of the $\epsilon$-net provides control on the trade-off between choice and number of landmarks and the quality of the approximate simplicial representation. 

We present three algorithms for constructing an $\epsilon$-net of a graph. We comparatively and empirically evaluate the efficiency and effectiveness of the choice of landmarks that they induce for the topological data analysis of different real-world graphs. 
\end{abstract}

\input{introduction.tex}

\input{relatedworks.tex}
\input{enet.tex}
\input{algorithms.tex}
\input{results.tex}

\input{conclusion.tex}

\bibliographystyle{splncs03}
\bibliography{biblio} 
\end{document}

%% file: introduction.tex
\section{Introduction}\label{sec:intro}
\textbf{Topological data analysis} (TDA)~\cite{carlssontda,otter2017roadmap} involves computation of topological features of datasets, such as persistent homology classes, and the representation of these topological features using such topological descriptors as persistence barcodes~\cite{edelsbrunner2010computational}.
In this section, we elaborate the computational blocks of topological data analysis as shown in Figure~\ref{fig:pipeline}. 

\textbf{Simplicial Complex.} Topological data analysis computes the topological features of a dataset, such as persistent homology classes, by computing the topological objects called \emph{simplicial complex}. A \textbf{simplicial complex} is constructed using simplices. Formally, a \emph{$k$-simplex} is the convex-hull of $(k+1)$ data points. For instance, a 0-simplex $[v_0]$ is a single point, a 1-simplex $[v_0v_1]$ is an edge, and a 2-simplex $[v_0v_1v_2]$ is a filled triangle. 
A \textbf{$k$-homology class} is an equivalent class of such $k$-simplicial complexes that cannot be reduced to a lower dimensional simplicial complex~\cite{edelsbrunner2010computational}.

In order to compute the $k$-homology classes, a practitioner does not have direct access to the underlying space of the point cloud and it is combinatorially hard to compute the exact simplicial representation of \v{C}ech complex~\cite{zomorodian2010fast}.
Thus different approximations of the exact simplicial representation are proposed: \emph{Vietoris-Rips} complex~\cite{vietoris1927hoheren} and \emph{lazy witness} complex~\cite{de2004topological}.

\begin{figure}[t!]
    \centering
    \includegraphics[width=0.8\textwidth]{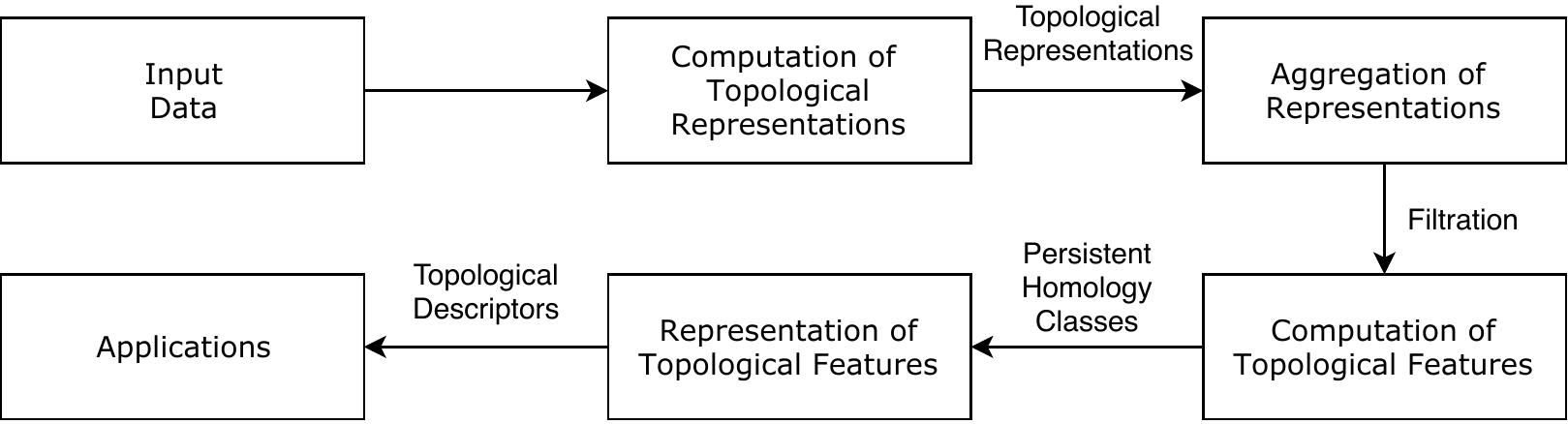}\vspace{-.5em}
    \caption{Components of topological data analysis.}\vspace{-3em}
\end{figure}\label{fig:pipeline}

\textbf{Approximate Simplicial Representations.} 
The \textbf{Vietoris-Rips complex} $R_\alpha(D)$, for a given dataset $D$ and real number $\alpha > 0$, is an abstract simplicial complex representation consisting of such $k$-simplices, where any two points $u, v$ in any of these $k$-simplices are at distance at most $\alpha$. Vietoris-Rips complex is the best possible ($\sqrt{2}$-)approximation of the \v{C}ech complex, computable with present computational resources, and is extensively used in topological data analysis literature~\cite{otter2017roadmap}. Thus, we use the Vietoris-Rips complex as the \emph{baseline representation} in this paper.
In the worst case, the number of simplices in the Vietoris-Rips complex grows exponentially with the number of data points~\cite{zomorodian2010fast}. 

Lazy witness complex~\cite{de2004topological} approximates the Vietoris-Rips complex by constructing the simplicial complexes over a subset of data points $L$, referred to as the landmarks. 
Formally, given a positive integer $\nu$ and a real number $\alpha>0$, the \textbf{lazy witness complex} $LW_\alpha(D,L,\nu)$ of a dataset $D$ is a simplicial complex over a landmark set $L$ where for any two points $v_i,v_j$ of a $k$-simplex $[v_0 v_1\cdots v_k]$, there is a point $w$ whose $(d^\nu(w) + \alpha)$-neighbourhood contains $v_i, v_j$. $d^\nu(w)$ is the  geodesic distance from point $w \in L$ to its $\nu$-th nearest point in the landmark set $L$.
In the worst case, the size of the lazy witness complexes grows exponentially with the number of landmarks. Less number of landmarks facilitates computational acceleration while produces a bad approximation of Vietoris-Rips with loss of topological features. Thus, the trade-off between the approximation of topological features and available computational resources dictates the choice of landmarks. \emph{We provide a quantification on such loss of topological features that was absent in the literature.}

\textbf{Filtration and Representation of Topological Features.} As the value of filtration parameter $\alpha$ increases, new simplices arrive and the topological features, i.e. the homology classes, start to appear. Some of the homology classes merge with the existing classes in a subsequent simplicial complex, and some of them persist indefinitely~\cite{edelsbrunner2010computational}.
In order to capture the evolution of topological structure with scale, topological data analysis techniques construct a sequence of simplicial complex representations, called a \emph{filtration}~\cite{edelsbrunner2010computational}, for an increasing sequence of $\alpha$'s. 
In a given filtration, the persistence interval of a homology class is denoted by $[\alpha_{b},\alpha_{d})$, where $\alpha_{b}$ and $\alpha_{d}$ are the filtration values of its appearance and merging respectively. The persistence interval of an indefinitely persisting homology class is denoted as $[\alpha_{b},\infty)$.
Topological descriptors, such as barcodes~\cite{collins2004barcode}, persistence diagram~\cite{edelsbrunner2010computational}, and persistence landscapes~\cite{bubenik2015statistical}, represent persistence intervals in order to draw qualitative and quantitative inference about the topological features. Distance measures between persistent diagrams such as the q-Wasserstein and Bottleneck distance~\cite{edelsbrunner2010computational} are often used to draw quantitative inference.

\textbf{Graph Topological Data Analysis.} Topological data analysis (TDA)~\cite{carlssontda} traditionally considers a point cloud in a Euclidean space. In many situations, however, data is available in the form of a weighted graph. Conveniently, the vertices of the graph with the geodesic distance define a metric space. This metric space is amenable to topological data analysis. This fact does not depend on whether the graph is embeddable in a Euclidean space or not. Though the metric space induced by a graph provides a different structure than point clouds to investigate, this metric space bears the similar issues of scalable construction of representations as well as similar approximate representations, filtrations, and topological descriptors. This motivated us to exploit the generalisability of topological data analysis and to extend the $\epsilon$-net induced lazy witness complexes~\cite{naheed2019} to graphs.

\textbf{Our Contributions.}
We investigate the computation of persistent homologies on a weighted graph. In Section~\ref{sec:net_graph}, we present a lazy witness complex approach leveraging the notion of $\epsilon$-net that we adapt to weighted graphs and their geodesic distance to select landmarks. We show that the $\epsilon$ parameter of the $\epsilon$-net gives a control on the trade-off between choice and number of landmarks and the quality of the approximate simplicial representation. 

In \Cref{subsec:prop}, we prove that an $\epsilon$-net is an $\epsilon$-approximate representation of the point cloud with respect to the Hausdorff distance. We prove that the lazy witness complex induced by an $\epsilon$-net, as a choice of landmarks, is a $3$-approximation of the induced Vietoris-Rips complex.

In \Cref{sec:algos}, we present three algorithms, namely Greedy-$\epsilon$-net, Iterative-$\epsilon$-net, and SPTpruning-$\epsilon$-net, for constructing an $\epsilon$-net of a graph. In Section~\ref{sec:experiment}, we comparatively and empirically evaluate the efficiency and effectiveness of the choice of landmarks that they induce for the topological data analysis of several real-world graphs. 

In Section~\ref{sec:conc}, we summarise the findings and the future directions of research that $\epsilon$-net opens up for graph topological data analysis.

%% file: relatedworks.tex
\begin{section}{Related Works}
\textbf{Graph TDA.} Existing applications of TDA to graphs focus on characterizing networks using features computed from persistence homology classes. \cite{Carstens2013} and \cite{petri2013topological} computed persistence homology at dimension 0, 1, and 2 of the clique filtration to study weighted collaboration networks (size $\sim$36000) and weighted networks from different domains (size $\sim$54000) respectively. 
In biology domain, \cite{duman2018gene} clustered gene co-expression networks (size $\sim$400) based on distances between Vietoris-Rips persistence diagram computed on each network. 
\cite{lee2011discriminative} studied Vietoris-Rips filtration of the functional brain networks computed on $\sim$100 region of interests (points) in human brains with different clinical disorders. They focus on homology classes at dimension 0 and 1. A related line of work regarding the topology on graphs involves graphs derived as a representation of point-cloud data (e.g. the neighbourhood graph) and their usage in data clustering~\cite{chazal2013tomato} and inference of global topology from local information~\cite{l-tda}. 

\textbf{Approximate Simplicial Complexes.} 
Computational infeasibility of constructing the \v{C}ech complex and Vietoris-Rips complex motivates the development of approximate simplicial representations such as the lazy witness complexes, sparse-Rips complex~\cite{sparse_rips} and graph induced complex (GIC)~\cite{gic}.  

\textbf{Applications of $\epsilon$-net.} The concept of $\epsilon$-net is a standard concept in analysis and topology~\cite{heinonen2012lectures} originating from the idea of $(\delta, \epsilon)$-limits formulated by Cauchy. $\epsilon$-net are sets in a metric space that covers the whole space and are well-separated. 
Nets have been used in nearest-neighbour search~\cite{krauthgamer2004navigating}. 
\cite{guibas2008reconstruction} used $\epsilon$-net for manifold reconstruction.
Graph induced complex~\cite{gic} uses the cliques in the neighbourhood graph to construct simplcial complex over an $\epsilon$-net. 

\cite{har2006fast} proposed net-tree data structure to represent $\epsilon$-nets at all scales of $\epsilon$. Net-tree is used to construct approximate well-separated pair decompositions~\cite{har2006fast} and approximate geometric spanners~\cite{har2006fast}. Sparse-Rips filtration~\cite{sparse_rips} constructs a net-tree on the point-cloud to decide which neighbouring points to delete.  
Contrary to Sheehy~\cite{sparse_rips}, we use $\epsilon$-net to select a fixed subset of points, called landmarks, and compute persistent homology using them.
\end{section}

%% file: enet.tex
\section{$\epsilon$-net of Graphs: Definition and Analysis}
\label{sec:net_graph}
\cite{naheed2019} proposes the $\epsilon$-net induced lazy witness complex for a point cloud embedded in a Euclidean space for efficient computation of topological data analysis. In practice, the datasets may not be represented as a point cloud in a Euclidean space. The data may have different representations and non-Euclidean geometry. For instance, the dataset with contextual and relational structure is often represented using graphs. In a graphical representation of data, the vertices represent the data objects, edges represent relations among the data objects, and weights on the edges quantify the amplitude of the relation with respect to others.

In this paper, we study both weighted and unweighted simple graphs.
Since unweighted graphs are a special case of weighted graph, we construct the definitions for weighted simple graphs. A weighted simple graph~\cite{berge} $G(V,E,W)$ is a graph with a vertex set $V$, an edge set $E$, a weight function $W:V \times V \rightarrow \mathbf{R}^{+}$, and does not contain any self-edge or multiple edge. The geodesic distance $d_G(u,v)$ between a pair $u,v$ of vertices in a graph is defined as the length of the shortest path between $u$ and $v$, where the path length is defined as the sum of weights of the edges connecting the vertices $u$ and $v$~\cite{newman2004analysis}. 
In this paper we treat a graph $G = (V,E,W)$ as a set $V$ endowed with the canonical metric $d_G: V \times V \rightarrow \mathbf{R}^{+}$. Thus, a weighted simple graph $G$ transforms into a metric space $(V,d_G)$.

As we substitute the points in the point cloud with the vertices of the graph and the Euclidean distance between points with the geodesic distance between the vertices, we adapt the components of topological data analysis for simple weighted graphs.

$\epsilon$-cover is a construction used in topology to compute inherent properties of a given space~\cite{heinonen2012lectures}.
In this paper, we import the concept of $\epsilon$-cover to define $\epsilon$-net of a graph. We use the $\epsilon$-net of a graph as landmarks for constructing lazy witness complex on that graph. 

We show that $\epsilon$-net, as a choice of landmarks, has guarantees such as being an $\epsilon$-approximate representation of the graph, its induced lazy witness complex being a $3$-approximation of the corresponding Vietoris-Rips complex, and also bounding the number of landmarks for a given $\epsilon$. These guarantees are absent for the other existing landmark selection algorithms such as random and maxmin algorithms~\cite{silva2006selecting}.

\vspace{-1em}
\subsection{Defining $\epsilon$-net of a Graph}\vspace{-1em}
In this paper, we consider a graph $G = (V,E,W)$ as a finite metric space $(V,d_G)$. Since there are many graphs that are neither Euclidean nor they have Euclidean embedding, we extend the definitions of neighbourhood and cover from point-set topology as follows before defining $\epsilon$-net. 
\begin{definition}[$\epsilon$-neighbourhood]\label{def:eneighbour}
	The $\epsilon$-neighbourhood $\mathcal{N}_\epsilon(u)$ of a vertex $u \in V$ is a subset of the vertex set $V$ such that for any vertex $v \in \mathcal{N}_\epsilon(u)$ the distance $d_G(u,v) \leq \epsilon$.
\end{definition}
The notion of $\epsilon$-cover for graph generalises the geometric notion of cover using the set-theoretic notion of $\epsilon$-neighbourhood in Definition~\ref{def:eneighbour}. 
\begin{definition}[$\epsilon$-cover]
	An $\epsilon$-cover of $G$ is the finite collection $\{\mathcal{N}_{\epsilon/2}(u_i) \}$ of $\epsilon/2$-neighbourhoods of vertices in $G$ such that $\cup_i \mathcal{N}_\frac{\epsilon}{2}(u_i) = V$.
\end{definition}
By triangle inequality, any set in the $\epsilon$-cover has the property that two vertices $u,v$ in the cover have geodesic distance $d_G(u,v)\leq \epsilon$. We differentiate an $\epsilon$-cover from the set of vertices whose $\frac{\epsilon}{2}$-neighbourhood determines that cover by defining $\epsilon$-sample. 
\begin{definition}[$\epsilon$-sample]
A set $L = \{u_1,u_2,\dots,u_{|L|}\} \subseteq V$ is an $\epsilon$-sample of graph $G$ if the collection $\{\mathcal{N}_{\epsilon}(u_i):u_i\in L\}$ of $\epsilon$-neighbourhoods covers $G$ i.e. $\cup_i \mathcal{N}_\epsilon(u_i) = V$.
\end{definition}
$\epsilon$-neighbourhoods may intersect, which is not an intended property if we want to decrease the size of the $\epsilon$-sample. We combine the notion of $\epsilon$-cover with the notion of $\epsilon$-sparsity to define $\epsilon$-net.
\begin{definition}[$\epsilon$-sparse]
A set $L = \{u_1,u_2,\dots,u_{|L|}\} \subset V$ is $\epsilon$-sparse if for any distinct $u_i,u_j \in L$,  $d_G(u_i,u_j) > \epsilon$ in graph $G$.
\end{definition}
An $\epsilon$-net of graph $G$ is such a subset of $V$ which is both an $\epsilon$-sample of $G$ and $\epsilon$-sparse. The $\epsilon$-net while considered as the landmark set $L$ induces a metric subspace $(L,d_L)$ of the metric space of the graph of $(V,d_G)$, where $d_L$ is the metric induced on set $L$ by the geodesic metric $d_G$.
\begin{definition}[$\epsilon$-net] 
A subset $L \subset V$ is an $\epsilon$-net of $G$ if $L$ is $\epsilon$-sparse and an $\epsilon$-sample of $G$.
\end{definition}
\textbf{Relating $\epsilon$-net and Other Graph Theoretic Concepts.}
The definition of $\epsilon$-net on graphs generalises the notion of independent set and dominating set for undirected graphs~\cite{berge}. Any $1$-net of an undirected graph $G = (V,E)$ is an independent set of $G$ and vice versa.  Any minimal cardinality $1$-net of $G$ is a dominating set of $G$, and vice versa.\vspace*{-1.3em}


\subsection{Analysing Properties of $\epsilon$-net of a Graph}
\label{subsec:prop}
$\epsilon$-net of a simple, weighted, connected graph comes with approximation guarantees irrespective of its algorithmic construction.
In this section, we provide following three analysis of $\epsilon$-net:
\begin{enumerate}
    \item An $\epsilon$-net of a connected graph is an $\epsilon$-approximation of its set of vertices in Hausdorff distance.
    \item The lazy witness complex induced by an $\epsilon$-net is a $3$-approximation of the Vietoris-Rips complex induced by the same set.
    \item For a graph of diameter $\Delta$, there exists an $\epsilon$-net of  of size at most $(\frac{\Delta}{\epsilon})^{O(\log{\frac{\abs{V}}{\epsilon}})}$.
\end{enumerate} 

\paragraph{Graph Approximation Guarantee of an $\epsilon$-net.}\label{subsec:hausdorff}
We use Lemma \ref{lemma:esampleg} to prove that the $\epsilon$-net of a graph $G = (V,E,W)$ is an $\epsilon$-approximation of $V$ in Hausdorff metric (Theorem \ref{thm:hausdorff}). Lemma \ref{lemma:esampleg} follows from the $\epsilon$-sample property of an $\epsilon$-net.
\begin{lemma}\label{lemma:esampleg}
Let $L$ be an $\epsilon$-net of graph $G$. For any vertex $v \in V$, there exists a point $u \in L \subseteq V$ such that the geodesic distance $d_G(u,v) \leq \epsilon$
\end{lemma}
\begin{proof}
Since $V = \cup_{u \in L}\mathcal{N}_\epsilon(u)$, for any vertex $v \in V$ there exists an $u \in L$ such that $v \in \mathcal{N}_\epsilon(u)$. 
As $v \in \mathcal{N}_\epsilon(u)$, by definition of $\epsilon$-neighbourhood, the length of the shortest path from $v$ to $u$ is at most $\epsilon$, i.e. $d_G(u,v) \leq \epsilon$.
\end{proof}


\begin{theorem}
\label{thm:hausdorff}
The Hausdorff distance between  $(V,d_G)$ and its $\epsilon$-net induced subspace $(L,d_L)$ is at most $\epsilon$.
\end{theorem}
\begin{proof}
 For any $u \in L \subseteq V$, there exists a vertex $v \in V$ such that $d_G(u,v) \leq \epsilon$, by definition of $\epsilon$-neighbourhood. Hence,
$\max_{L}{\min_{V}{d_G(u,v)}} \leq \epsilon$. 
By Lemma~\ref{lemma:esampleg},  $\max_{V}{\min_{L}{d_G(u,v)}} \leq \epsilon$. Since the Hausdorff distance $d_H(V,L)$  is defined as the maximum of $\max_{L} \min_{V}{d_G(u,v)}$ and $\max_{V}{\min_{L}{d_G(u,v)}}$, therefore $d_H(V,L)$ is upper bounded by $\epsilon$.
\end{proof}

\paragraph{Topological Approximation Guarantee of an $\epsilon$-net induced Lazy witness complex on graphs.} 
\label{subsec:theory}
In addition to an $\epsilon$-net being an $\epsilon$-approximation of the space $(V,d_G)$, we prove that the lazy witness complex induced by the $\epsilon$-net, as landmarks, is a good approximation (\Cref{thm:approx}) to the Vietoris-Rips complex on the same set of vertices. This approximation ratio is independent of the algorithm constructing the $\epsilon$-net. As a step towards \Cref{thm:approx}, we state \Cref{lemma:nn1}. \Cref{lemma:nn1} is implied by the definition of the lazy witness complex and $\epsilon$-sample. \Cref{lemma:nn1} establishes the relation between 1-nearest neighbour of points in an $\epsilon$-net.

\begin{lemma}\label{lemma:nn1}
If $L$ is an $\epsilon$-net of $(V,d_G)$, the distance $d_G(u,v)$ from any vertex $u \in L$ to its 1-nearest neighbour $v \in V$ is at most $\epsilon$.
\end{lemma}
\Cref{thm:approx} shows that a lazy witness complex induced by an $\epsilon$-net landmarks is a $3$-approximation of the Vietoris-Rips complex on the landmarks beyond a certain value of the filtration parameter. 

\begin{theorem}
	\label{thm:approx}
	If $L$ is an $\epsilon$-net of the point cloud $V$ for $\epsilon \in \mathbf{R}^+$, $LW_{\alpha}(V,L,\nu = 1)$ is the lazy witness complex of $L$ at filtration value $\alpha$, and $R_\alpha(L)$ is the Vietoris-Rips complex of $L$ at filtration $\alpha$, 
	$R_{\alpha/3}(L) \subseteq LW_\alpha(V,L,1) \subseteq R_{3\alpha}(L)$ for $\alpha \geq 2\epsilon$.
\end{theorem}\vspace{-1.5em}
\begin{proof}	
	In order to prove the first inclusion, consider a $k$-simplex $\sigma_k = [x_0 \cdots x_k]$ $\in R_{\alpha/3}(L)$. For any edge $[x_ix_j] \in \sigma_k$, let $w_t$ be the point in $V$ that is nearest to the vertices of $[x_ix_j]$. Without loss of generality, let that vertex be $x_j$. 
	Since $w_t$ is the nearest neighbour of $x_j$, by Lemma~\ref{lemma:nn1}, $d_G(w_t,x_j) \leq \epsilon \leq \frac{\alpha}{2}$. 
	Since $[x_ix_j] \in  R_{\alpha/3}$, $d_G(x_i,x_j) \leq \frac{\alpha}{3} < \frac{\alpha}{2}$. By triangle inequality, $d_G(w_t,x_i) \leq \frac{\alpha}{2} +  \frac{\alpha}{2} \leq \alpha$. Hence, $x_i$ is within distance $\alpha$ from $w_t$. The $\alpha$-neighbourhood of point $w_t$ contains both $x_i$ and $x_j$. 
	Since $d^1(w_t) > 0$, the $(d^1(w_t) + \alpha)$-neighbourhood of $w_t$ also contains $x_i,x_j$. Therefore, $[x_ix_j]$ is an edge in $LW_\alpha(V,L,1)$. Since the argument is true for any $x_i,x_j \in \sigma_k$, the $k$-simplex $\sigma_k \in LW_\alpha(V,L,1)$.
	
	In order to prove the second inclusion, consider a $k$-simplex $\sigma_k = [x_0x_1 \cdots x_k]$ $\in LW_\alpha(V,L,1)$. Therefore, by definition of lazy witness complex, for any edge $[x_ix_j]$ of $\sigma_k$ there is a witness $w \in V$ such that, the $(d^1(w) + \alpha)$-neighbourhood of $w$ contains both $x_i$ and $x_j$. Hence, $d_G(w,x_i) \leq d^1(w) + \alpha \leq \epsilon + \alpha$  (by Lemma ~\ref{lemma:nn1})$\leq 3\alpha/2$. Similarly, $d_G(w,x_j) \leq 3\alpha/2$. By triangle inequality, $d_G(x_i,x_j) \leq 3\alpha$. Therefore, $[x_ix_j]$ is an edge in $R_{3\alpha}(L)$. Since the argument is true for any  $x_i,x_j \in \sigma_k$, the k-simplex $\sigma_k \in R_{3\alpha}(L)$.  
\end{proof}

\emph{Discussion.} Theorem~\ref{thm:approx} implies that the interleaving of lazy witness filtration $LW = {LW_\alpha(L)}$ and the Vietoris-Rips filtration $R = R_\alpha(L)$ occurs when $\alpha > 2\epsilon$. As a consequence, their corresponding partial persistence diagrams $Dgm_{>2\epsilon}(LW)$ and $Dgm_{>2\epsilon}(R)$ are $3\log{3}$-approximations of each other in log-scale, by the persistence approximation lemma~\cite{sparse_rips}. In \Cref{sec:experiment}, we empirically validate this bound for the lazy witness complex induced by the $\epsilon$-net landmarks.

\paragraph{Size of an $\epsilon$-net.} 
We prove an upper bound on the size of an $\epsilon$-net of a connected unweighted graph using the \emph{doubling dimension}.

The doubling dimension~\cite{gupta2003bounded}  of a metric space $M = (X,d)$ is the smallest positive number $D$ such that any $\epsilon$-neighbourhood in $M$ can be covered by $2^D$ number of $\frac{\epsilon}{2}$-neighbourhoods. 
A metric space is called \emph{doubling} if its doubling dimension is bounded. The space $(V,d_G)$ is a doubling metric space.



Gupta et. al.~\cite{gupta2003bounded} showed that the doubling dimension $D(G)$ of an unweighted connected graph $G$ is related to its local density. The local density of an unweighted connected graph $G$, denoted $\beta(G)$, is the smallest value $\beta$ such that $\abs{\mathcal{N}_\epsilon(v)} \leq \beta \epsilon$, for all $v \in V$ and $\epsilon \in \mathbf{N}$. To be precise the doubling dimension $D(G) \leq 4 \log(2\beta(G))$~\cite{gupta2003bounded}. We use this result along with the following lemma to prove bound on the size of an $epsilon$-net of an unweighted graph in Theorem~\ref{thm:epsnetsz}.

\begin{lemma}
\label{lemma:doubling}
For any connected unweighted graph $G$ of diameter $\Delta$ and doubling dimension $D(G)$, there exists an $\epsilon$-net of size at most $(\frac{\Delta}{\epsilon})^{D(G)}$ where $\Delta \geq \epsilon \geq 1$
\end{lemma}\vspace{-1.5em}
\begin{proof}
Let $\Delta$ be the diameter of an unweighted graph $G$ of doubling dimension $D(G)$. Thus, $V$ (a $\Delta$-neighbourhood) can be covered by $2^{D(G)}$ number of $\Delta/2$-neighbourhoods. Each $\Delta/2$-neighbourhood can be covered by $2^{D(G)}$ number of $\Delta/4$-neighbourhoods. Thus $V$ can be covered by $2^{2D(G)}$ number of $\Delta/4$-neighbourhoods. Repeating this $\log_2(\frac{\Delta}{\epsilon})$-times, we get that $V$ can be covered by $2^{D(G)\log_2(\frac{\Delta}{\epsilon})}$ number of $\epsilon$-neighbourhoods. Each of the $\epsilon$-neighbourhoods contain at most one $\epsilon$-net-landmark. Hence, there exists an $\epsilon$-net of size $(\frac{\Delta}{\epsilon})^{D(G)}$ for any connected, unweighted graph $G$.
\end{proof}\vspace{-1.5em}
\input{algo_1.tex}

\begin{theorem}
\label{thm:epsnetsz}
    For any connected unweighted graph $G = (V,E)$ of diameter $\Delta$, there exists an $\epsilon$-net of size at most $(\frac{\Delta}{\epsilon})^{O(\log(\frac{\abs{V}}{\epsilon}))}$ 
\end{theorem}\vspace{-1.5em}
\begin{proof}
In a connected unweighted graph $G$, the size $\abs{\mathcal{N}_\epsilon(v)}$ of $\epsilon$-neighbourhood of a vertex $v$ is greater than $\epsilon$ for $\Delta \geq \epsilon \geq 1$.
Thus, local density $\beta(G) > 1$. We observe that $\max_v \abs{\mathcal{N}_\epsilon(v)} < \abs{V}$. Thus, $\beta(G) = \max_{v,\epsilon} \frac{\abs{\mathcal{N}_\epsilon(v)}}{\epsilon}$ is at most $\frac{\abs{V}}{\epsilon}$. Applying the result from Gupta et.al.~\cite{gupta2003bounded}, the doubling dimension $D(G) \leq 4 \log(\frac{\abs{V}}{\epsilon})$. The rest  follows from~\Cref{lemma:doubling}.

\end{proof}

%% file: algo_1.tex
\setlength{\textfloatsep}{6pt}
\begin{algorithm}[t!]
\caption{Greedy-$\epsilon$-Net}
\label{alg:base}
\begin{algorithmic}[1]
\Require Graph $G = (V,E,W)$, parameter $\epsilon$
\Ensure Set of Landmarks $L$
\State Initialize $L = \phi$
\State Let $nc$ be the hash table with vertices as keys, number of vertices in their $\epsilon$-cover as value.
\ForAll{$u \in V$}
    \State $nc[u] =$ \textbf{$\epsilon$-BFS$(G,u,\epsilon)$}
\EndFor
\State Initialize all vertex $u \in V$ as marked.
\Repeat
    \State Sort $nc$ in descending order of its value. 
    \If{ $marked[u]$ = False}
        \State $L.insert(u)$
        \ForAll{vertex $u'$ in u's $\epsilon$-cover}
            \State Mark $u'$ as True.
            \State Delete key $u'$ from nc.
            \ForAll{$v \in V$}
                \If{$u'$ is in $\epsilon$-cover of $v$}
                    \State Decrease $nc[v]$ by 1.
                \EndIf
            \EndFor
        \EndFor
    \EndIf
    \State Delete key $u$ from nc.
\Until{all vertices are marked} 
\end{algorithmic}
\end{algorithm}
\begin{algorithm}[t!]
\caption{$\epsilon$-BFS}
\label{alg:ebfs}
\begin{algorithmic}[1]
\Require Graph $G = (V,E,W)$, vertex $u$, parameter $\epsilon$.
\Ensure Set of vertices in $u$'s $\epsilon$-cover, $C_\epsilon$
\State Initialize Queue $Q = \{u\}$
\State initialize $C_\epsilon = \phi$
 \While{$Q \neq \phi$} 
    \State $v = \text{DEQUEUE}(Q)$
    \State $v.marked$ = True
    \ForAll{ $v' \in G.Adj[v]$}
        \If{ $v'.marked$ = False}
            \State $v'.d = v.d + W[v,v']$
            \If{ $v'.d \leq \epsilon$}
                \State $C_\epsilon = C_\epsilon \cup \{v'\}$
            \EndIf
        \EndIf
    \EndFor
 \EndWhile
\end{algorithmic}
\end{algorithm}

%% file: algorithms.tex
\input{Algos.tex}

%% file: Algos.tex
\section{Algorithms for Computing $\epsilon$-Net}\label{sec:algos}
In this section, we propose and elaborate three algorithms, namely Greedy-$\epsilon$-Net, Iterative-$\epsilon$-Net, and SPTprunning-$\epsilon$-net, for computing $\epsilon$-net on graphs. 
\vspace{-1em}
\subsection{Greedy-$\epsilon$-Net Algorithm}\vspace{-1em}
We propose a greedy algorithm, namely Greedy-$\epsilon$-Net, to compute a minimal-cardinality $\epsilon$-net of a graph. Greedy-$\epsilon$-Net (\Cref{alg:base}) maintains a hash table with vertices as keys and the number of vertices in their $\epsilon$-cover as values. At each step, Greedy-$\epsilon$-Net selects a vertex with the largest $\epsilon$-cover, marks the covered vertices, and updates the $\epsilon$-cover of other vertices until all the vertices are marked as covered.

\vspace{-1em}
\subsection{Iterative-$\epsilon$-Net Algorithm}\vspace{-1em}
Iterative-$\epsilon$-Net (\Cref{alg:idif}) is a diffusive algorithm that maintains a set $C_{(\epsilon,2\epsilon)}$ that contains the set of unmarked vertices that are within a ring of $(\epsilon,2\epsilon]$ distance from the current set of landmarks. We call them the \textbf{ring vertices}.
Iterative-$\epsilon$-Net also maintains another set $C_{2\epsilon}$ that contains the set of unmarked vertices that are at distance at least $2\epsilon$ but are adjacent to the bordering vertices of the cover. We call them the \textbf{enveloping vertices}.

Iterative-$\epsilon$-net, at each iteration, uniformly at random selects a vertex $u$ from the current set of enveloping vertices as next landmark, and run Partial-BFS (\Cref{alg:pbfs}) starting at $u$ to mark the vertices in its $\epsilon$-cover, as well as to update the enveloping and ring vertex sets. If enveloping vertex set is empty it selects a ring vertex uniformly at random as next landmark.

Iterative $\epsilon$-net algorithm has the property that some vertex in the $\epsilon$-cover of landmark $l_{i+1}$ is always adjacent to some vertex in the $\epsilon$-cover of landmark $l_i$. Thus the two covers are adjacent as sets.\vspace*{-1pt}


\begin{algorithm}[t!]
\caption{Iterative-$\epsilon$-Net}
\label{alg:idif}
\begin{algorithmic}[1]
\Require Graph $G = (V,E,W)$, parameter $\epsilon$
\Ensure Set of landmarks $L$
\State Initialise $L = \phi$
\State $i = 1$
\State Select initial landmark $l_i= u$ uniformly at random from $V$.
\State $L = L \cup \{u\}$
\State Let $C^{u}_{(\epsilon,2\epsilon)}$  be the set of unmarked vertices $v$ such that $\epsilon < d_G(u,v) \leq 2\epsilon$.
\State Let $C^{u}_{2\epsilon}$ be the set of vertices $v$ such that $d_G(u,v)>2\epsilon$ and are the closest.
 \Repeat 
    \State $C^{u}_{\epsilon,2\epsilon} = \textbf{PartialBFS}(G,l_i,C^{u}_{(\epsilon,2\epsilon)},C^{u}_{2\epsilon})$ 
    \If{$C^{u}_{2\epsilon}$ is empty}
        \State select $l_{i+1}$ uniformly at random from $C^{u}_{(\epsilon,2\epsilon)}$.
    \Else
        \State select $l_{i+1}$ uniformly at random from $C^u_{2\epsilon}$.
    \EndIf
    \State $L = L \cup \{l_{i+1}\}$.
    \State $i = i + 1$
    \State $u = l_{i+1}$
 \Until{all vertices are marked}
\end{algorithmic}
\end{algorithm}
\begin{algorithm}[t!]
\caption{PartialBFS}
\label{alg:pbfs}
\begin{algorithmic}[1]
\Require Graph $G$, vertex $u$, set $C^{u}_{(\epsilon,2\epsilon)}$, set $C^{u}_{2\epsilon}$.
\State Initialise Queue $Q = \{u\}$
\State $u.marked$ = True

 \While{$Q \neq \phi$} 
    \State $v = \text{DEQUEUE}(Q)$
    \ForAll{ $v' \in G.Adj[v]$}
    \If{ $v'.marked$ = False}
        \State $v'.d = v.d + W[v,v']$
        \If{ $v'.d \leq \epsilon$}
            \State $v'.marked$ = True
            \State Remove $v'$ from $C^{u}_{(\epsilon,2\epsilon)}$ and $C^{u}_{2\epsilon}$ if exists.
            \State $\text{ENQUEUE}(Q,v')$
        \ElsIf{$\epsilon< v'.d \leq 2\epsilon$}
            \State $C^{u}_{(\epsilon,2\epsilon)} = C^{u}_{(\epsilon,2\epsilon)} \cup \{v'\}$
            \State $\text{ENQUEUE}(Q,v')$
        \Else
            \State $C^{u}_{2\epsilon} = C^{u}_{2\epsilon} \cup \{v'\}$
        \EndIf
    \EndIf
    
    \EndFor
 \EndWhile
\end{algorithmic}
\end{algorithm}
\vspace*{-.5em}
\subsection{Sparsity based Pruning Algorithm on Shortest Path Tree}\vspace*{-.5em}
We propose SPTpruning-$\epsilon$-Net algorithm (Algorithm~\ref{alg:treeeps}) that constructs a shortest path tree of the graph and uses the tree to compute an $\epsilon$-net.
Algorithm~\ref{alg:treeeps} computes a shortest path tree rooted at a vertex chosen uniformly at random. The algorithm uses $\epsilon$-BFS (Algorithm~\ref{alg:ebfs}) to construct a preliminary BFS spanning tree of the graph (line 2). Then the algorithm constructs an $\epsilon$-net of the BFS-tree (line 4-28). It does so by traversing the tree level-order starting from root, running $\epsilon$-BFS in the tree to mark covered vertices, and add the set of unmarked vertices at level $\epsilon+1$  as candidates for landmarks. 

An $\epsilon$-net $L_{SPT}$ of a BFS-tree $SPT(G)$ of a graph $G$ has the property that, any vertex $v \in SPT(G) \subset G$ that is covered by some vertex $u \in L_{SPT}$, is also covered by the $u \in V$ in the graph as well. This property follows from the fact that, the distance between any vertex pair $d_G(u,v)$ in the graph can only be shorter than their distance $d_{SPT}(u,v)$ in the BFS-tree. 

Unless one of the vertex is a root, the distance between any pair of vertices in the tree is not guaranteed to be the shortest in the graph. Thus an $\epsilon$-net of a BFS-tree does not have $\epsilon$-sparsity in the graph containing the BFS-tree as a subgraph. As a remedy, Algorithm~\ref{alg:prune} prunes vertices from the candidate landmarks that are covered by other candidate landmark. 


\begin{algorithm}[t!]
\caption{SPTpruning-$\epsilon$-net}
\label{alg:treeeps}
\begin{algorithmic}[1]
\Require Graph $G$, Diameter $\Delta$, parameter $\epsilon$
\Ensure Set of landmarks $C$
\State Select a vertex $u$ uniformly at random from $V$.
\State Run \textbf{$\epsilon$-BFS(G,u,$\Delta$)} to construct a BFS spanning tree rooted at $u$
\State Let $SPT$ be the tree.
\State Initialise Queue $Q = \{u\}$
\State $u.marked$ = True
\State Let $\epsilon$-net of the SPT $C_\epsilon = \phi$ 
\Repeat 
 \While{$Q \neq \phi$} 
    \State $u' = \text{DEQUEUE}(Q)$
    \State u'.marked = True
    \State $u'.d = 0$
    \State Initialise Queue $Q' = \{u'\}$
    \While{$Q' \neq \phi$} 
        \State $v = \text{DEQUEUE}(Q')$
        \ForAll{ $v' \in SPT.Adj[v]$}
        \If{ $v'.marked$ = False}
            \State $v'.d = v.d + W[v,v']$
            \If{ $v'.d \leq \epsilon$}
                \State $v'.marked$ = True
                \State $\text{ENQUEUE}(Q',v')$
            \Else
                \State $C_\epsilon = C_{\epsilon} \cup \{v'\}$
                \State $\text{ENQUEUE}(Q,v')$
            \EndIf
        \EndIf
        
        \EndFor
    \EndWhile
 \EndWhile
 \Until{All vertex in $SPT$ are marked}
\State C = Prune(G,$C_\epsilon$,$\epsilon$)
\end{algorithmic}
\end{algorithm}
\begin{algorithm}[t!]
\caption{Prune}
\label{alg:prune}
\begin{algorithmic}[1]
\Require Graph $G$, A set of landmarks $C_\epsilon$, parameter $\epsilon$
\Ensure Set of landmarks $C$
\State $C = C_\epsilon$
\ForAll{vertex $v \in C_\epsilon$}
    \If {v.marked = False}
        \State Run \textbf{$\epsilon$-BFS(G,v,$\epsilon$)}
    \Else
        \State $C = C \setminus \{v\}$
    \EndIf
\EndFor
\end{algorithmic}
\end{algorithm}

%% file: results.tex
\section{Performance Analysis}\label{sec:experiment}\vspace{-.5em}
In this section, we experimentally and comparatively analyse the performance of the three proposed algorithms to compute an $\epsilon$-net of graphs. We discuss the effectiveness and efficiency of the algorithms and also validate that the $\epsilon$-net computed by any of these algorithms satisfy being $3$-approximation of the Vietoris-Rips complex.
\vspace{-1em}\subsection{Datasets}\vspace{-.5em}
We evaluate the performance of our algorithms using two real-world datasets. The dataset Power~\cite{konect:duncan98} is an unweighted graph of US Power-grid (4941 vertices, 6594 edges, diameter 46).
Celegans~\cite{badhwar2015control} is a weighted graph of Celegans worm's frontal neural network (297 vertices, 2148 edges, diameter 1.333). 
\vspace{-1em}
\subsection{Experimental Setup}\vspace{-.5em}
We implement the experimental workflow in C++. We use Snap library~\cite{leskovec2016snap} for graph processing and $\epsilon$-net computations. We modify the Ripser\Ripser library to compute lazy witness complexes and their persistent intervals.  We use R-TDA package~\cite{rTDA} to compute bottleneck distances. All experiments are run on a machine with an Intel(R) Xeon(R)@2.20GHz CPU and 80 GB memory limit. We set the lazy witness parameter $\nu = 1$ in all computations. We set the maximum value of the filtration parameter to the diameter of the corresponding dataset. We compute persistent intervals at dimension 0 and 1. 

 
\begin{figure}[t!]
    \centering
    \includegraphics[width=0.9\textwidth]{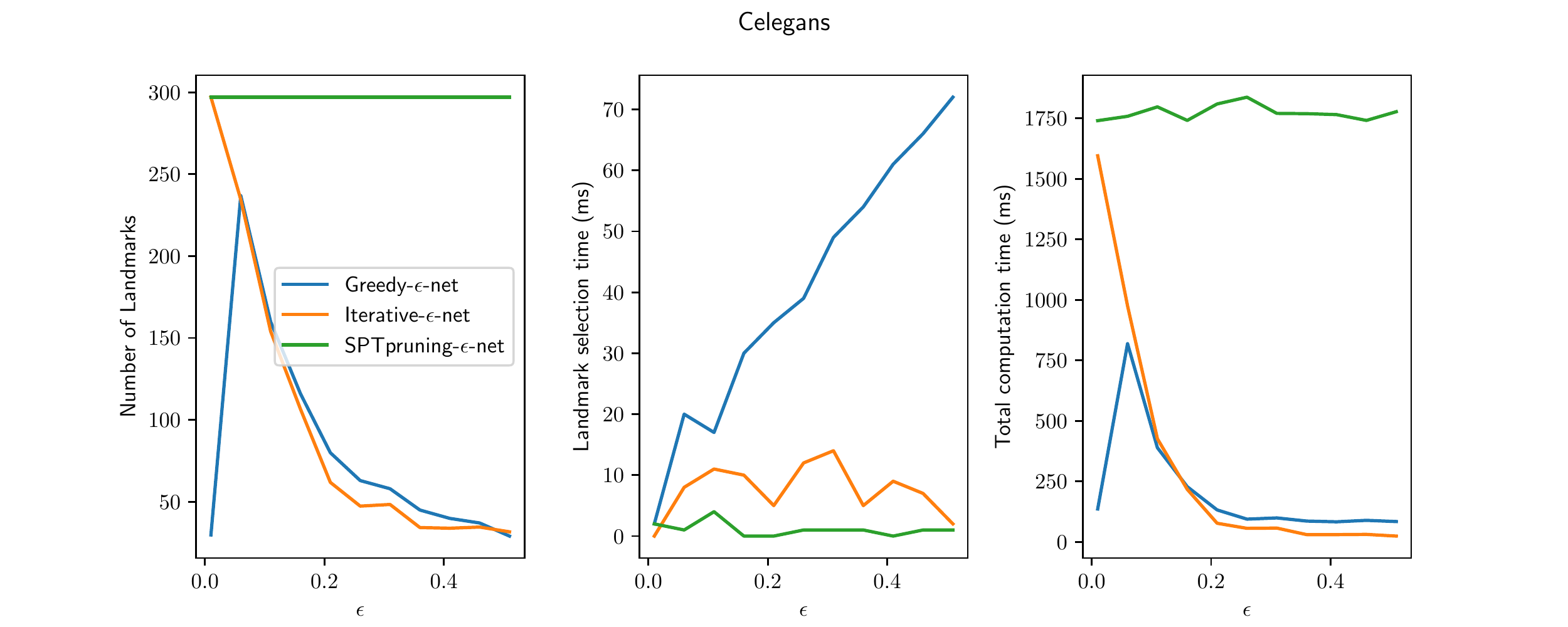}
    \caption{Performance of the algorithms on Celegans dataset.}
    \label{fig:effi_cel}
\end{figure}
\begin{figure}[t!]
    \centering
    \includegraphics[width=0.9\textwidth]{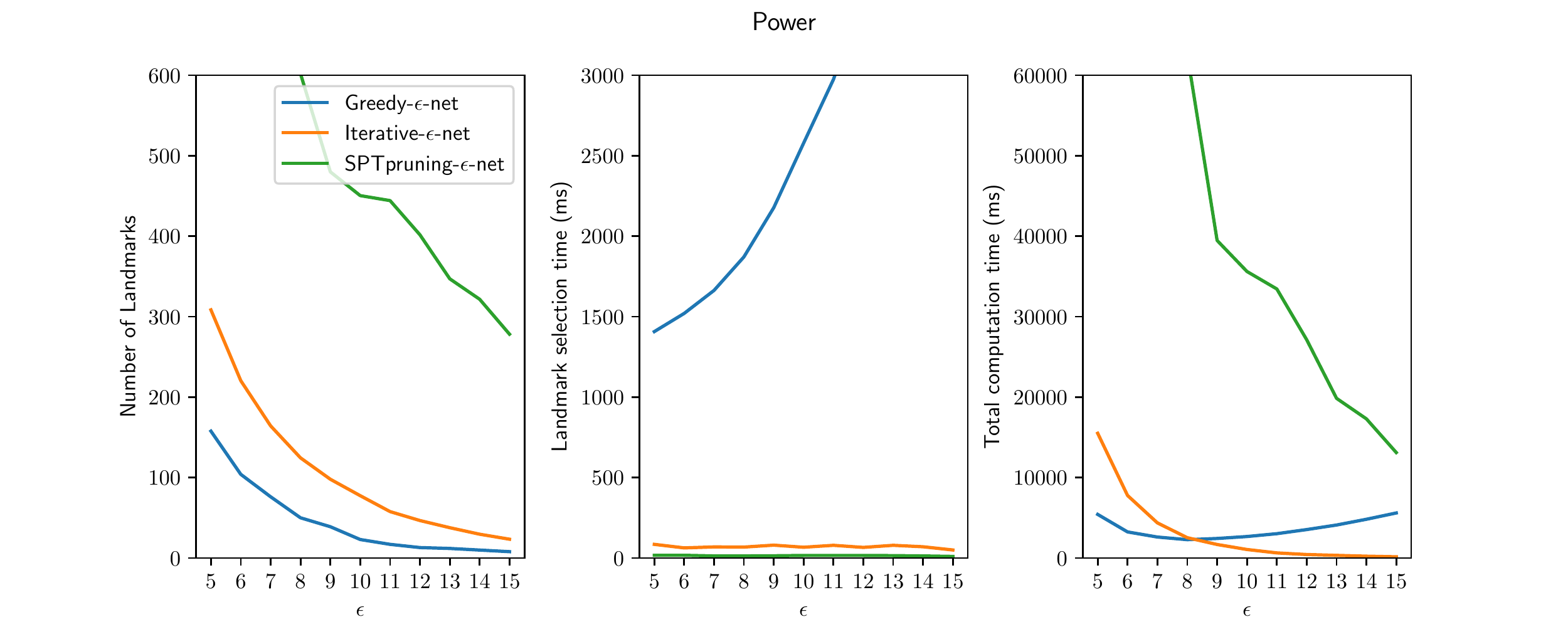}
    \caption{Performance of the algorithms on Power dataset.}
    \label{fig:effi_pow}
\end{figure}\vspace{-1em}
\subsection{Efficiency and Effectiveness}
We measure the efficiency of our algorithms in terms of CPU time (in ms) required to select $\epsilon$-net landmarks for a given $\epsilon$, and the overall computation of persistent homology of each graph. The overall total computation time includes the time an algorithm spends constructing $\epsilon$-net and time spent on computing persistent intervals at dimension 0 and 1. 
We measure the effectiveness of the algorithms using the number of landmarks they select by corresponding $\epsilon$-net construction.

Figures~\ref{fig:effi_cel} and~\ref{fig:effi_pow} illustrate the experimental results for the Celegans and the Power dataset respectively. We observe that the number of landmarks selected by all the algorithms decrease as the $\epsilon$-increases. 
The landmark selection time of the Greedy-$\epsilon$-Net increases as the $\epsilon$ increases independent of the dataset.
For other two algorithms, the landmark selection time varies depending on the density of the graph. 
The landmark computation time of Iterative-$\epsilon$-net and SPTpruning-$\epsilon$-net are almost invariant with $\epsilon$. 

We observe that Iterative-$\epsilon$-Net takes longer time ($>10$ ms) to select landmarks compared to the SPTpruning-$\epsilon$-Net algorithm but it selects less number of landmarks than SPTpruning-$\epsilon$-Net algorithm. Thus, the overall runtime of lazy witness complex construction using Iterative-$\epsilon$-Net is smaller than that of the SPTpruning-$\epsilon$-Net. The empirical performance analysis instantiates Iterative-$\epsilon$-Net as the practical and efficient choice to construct $\epsilon$-net induced lazy witness complex for graphs.

\begin{figure}[t!]
    \centering\vspace{.2em}
    \includegraphics[width=0.6\textwidth,trim={5cm 7.2cm 5cm 7.2cm}]{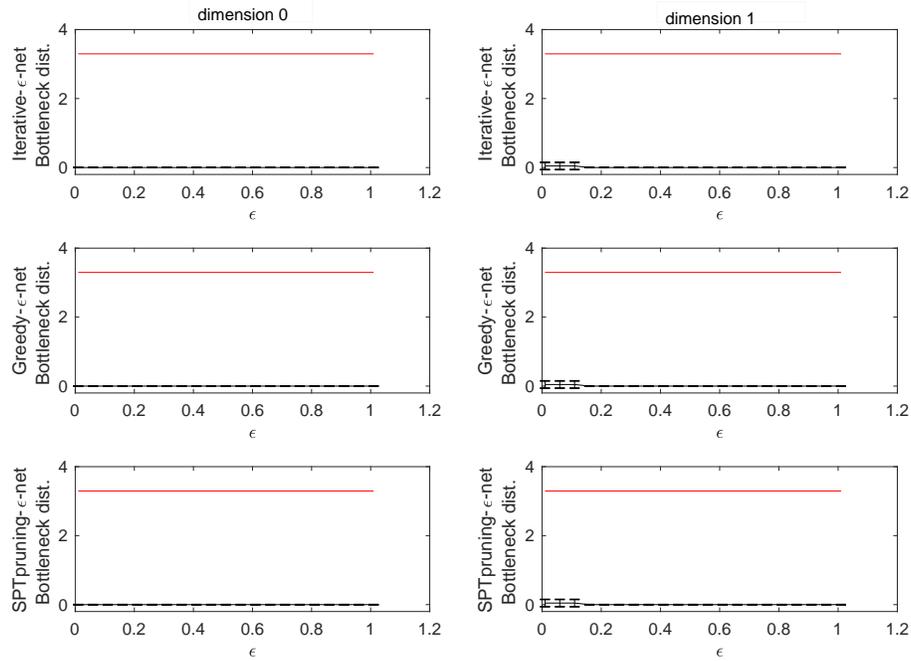}
    \caption{Validation of the approximation guarantee of the $\epsilon$-net induced lazy witness complexes for Celegans dataset. The bottlneck distance between the partial persistence diagram of the corresponding Vietoris-Rips and lazy witness filtrations is less than $3\log{3}$ in both dimensions 0 and 1.}
    \label{fig:val}
\end{figure}\vspace*{-2.5pt}


\subsection{Validating the Approximation Guarantee}\vspace*{-1em}
In order to validate the approximation guarantee in Theorem~\ref{thm:approx}, we construct the Vietoris-Rips complex and lazy witness complexes using $\epsilon$-nets for different values of $\epsilon$ and different algorithms. We compare the corresponding complexes by computing bottleneck distances between the persistence diagrams at dimension 0 and 1. We retain only the partial diagram with points in the diagram born after $2\epsilon$ for this purpose. Figure~\ref{fig:val} validates the guarantee on Celegans dataset. We omit the validation on Power dataset for the sake of brevity.

%% file: conclusion.tex
\section{Conclusion and Future Work}~\label{sec:conc}
We investigate the computation of persistent homologies on  weighted graphs. We extend the notion of $\epsilon$-net for point clouds to weighted graphs. We further propose an $\epsilon$-net induced lazy witness complex that leverages the $\epsilon$-net and their geodesic distances to select landmarks. 

We prove that an $\epsilon$-net of a connected graph is an $\epsilon$-approximation of its set of vertices in Hausdorff distance.
We also prove that the lazy witness complex induced by an $\epsilon$-net is a $3$-approximation of the Vietoris-Rips complex induced by the same set.
We prove the existence of an $\epsilon$-net (of a graph) of size at most $(\frac{\Delta}{\epsilon})^{O(\log{\frac{|V|}{\epsilon}})}$, where $\Delta$ is the diameter of the graph.
    
We present three algorithms for constructing an $\epsilon$-net of a graph. We comparatively and empirically evaluate the efficiency and effectiveness of the choice of landmarks that they induce for the topological data analysis of several real world graphs. The empirical performance analysis instantiates Iterative-$\epsilon$-Net as the practical and efficient choice to construct $\epsilon$-net induced lazy witness complex for graphs.

An interesting future work would be to leverage the notion of $\epsilon$-net induced lazy witness complex for faster and scalable computation of machine learning problems, such as clustering~\cite{chazal2013tomato}, deep learning~\cite{pmlr-v97-hofer19a}, kernel density estimation~\cite{carriere2017sliced}, for both graphs and point clouds.

\section*{Acknowledgement}
This work is partially supported by the National University of Singapore Institute for Data Science project WATCHA and by Singapore Ministry of Education project Janus.